\documentclass[12pt,draftcls,onecolumn]{IEEEtran}
\usepackage{amsfonts}
\usepackage{amsmath}
\usepackage{amssymb}
\usepackage{amsthm}
\usepackage{graphicx}
\usepackage{epsfig}

\numberwithin{equation}{section}
\newtheorem{thm}{Theorem}[section]
\newtheorem{prop}{Proposition}[section]
\newtheorem{lem}{Lemma}[section]
\newtheorem{rem}{Remark}[section]
\newtheorem{cor}{Corollary}[section]
\newtheorem{defn}{Definition}[section]

\newcommand{\thmref}[1]{Theorem~{\rm \ref{#1}}}
\newcommand{\lemref}[1]{Lemma~{\rm \ref{#1}}}

\newcommand{\propref}[1]{Proposition~{\rm \ref{#1}}}
\newcommand{\defnref}[1]{Definition~{\rm \ref{#1}}}

\def\one{{\hbox{1{\kern -0.35em}1}}}

\title{Consensus over a Random Network Generated by
  i.i.d. Stochastic Matrices} 

\author{ Qingshuo Song, \thanks{Department of Mathematics, City
    University of Hong Kong, 83 Tat Chee Avenue, Kowloon Tong, Hong
    Kong, {\tt
      song.qingshuo@cityu.edu.hk}.}
  \and Guanrong Chen, \thanks{Department of Electronic Engineering, City
    University of Hong Kong, 83 Tat Chee Avenue, Kowloon Tong, Hong
    Kong, {\tt gchen@ee.cityu.edu.hk}.}
  \and Daniel W. C. Ho
  \thanks{Department of Mathematics, City, University of Hong Kong, 83
    Tat Chee Avenue, Kowloon Tong, Hong Kong, {\tt
      madaniel@cityu.edu.hk}.}
}

\begin{document}

\maketitle
\begin{abstract}
  Our goal is to find a necessary and sufficient condition on the
  consensus over a random network, generated by i.i.d. stochastic
  matrices. We show that the consensus problem in three different
  convergence modes (almost surely, in probability, and in $L^1$) 
  are equivalent, thus have the same necessary
  and sufficient condition.  We obtain the necessary and sufficient
  condition through the stability in a projected subspace.

\end{abstract}

\section{Introduction}
We consider a stochastic linear difference equation,
$$X(t) = A(t) X(t-1), \quad t = 1, 2, \ldots$$
where the states $\{X(t)\}$ is an $\mathbb{R}^N$-valued sequence, and
$\{A(t)\}$ is a sequence of i.i.d. (independent and identically
distributed) right 
stochastic matrices (non-negative matrix with  each row summing to
$1$). The system is said to reach consensus if, for any initial state, 
$\max_{1\le i,j \le N} |X_i(t) - X_j(t)|$ converges to zero as $t \to
\infty$ in an appropriate sense. Since $X(t)$ is random, there
are  different modes of consensus: Almost surely consensus, in
probability consensus, and $L^1$ consensus.  

Consensus problem over a difference equation has wide applications in
random network theory (e.g., \cite{HM05, TJ08} and the references 
therein). \cite{HM05} studies the consensus  in
probability. \cite{TJ08} establish an elegant  necessary and sufficient
condition  for almost surely consensus by investigating the ergodicity
of a random matrix sequence:
\begin{equation}
  \label{eq:nscond}
  |\lambda_2(\mathbb{E} [A(1)])| <1
\end{equation}
where $\mathbb{E}[\cdot]$ is the expectation operator and
$\lambda_2(\cdot)$ is the second largest eigenvalue (in absolute
value) of the argument matrix. 
Note that, almost surely consensus implies in probability consensus,
and thus \eqref{eq:nscond} is obviously a sufficient condition for in
probability consensus.   

In this work, looking further into the specific nature of
$\{X(t)\}$, we show that the consensus in all three modes are actually
indifferent, hence  \eqref{eq:nscond}  gives necessary and  
sufficient condition for consensus in all three modes. In addition, by
using a completely different methodology in contrast to \cite{TJ08},
our result applies to a more general setting: their restriction on
the space of stochastic matrices with strictly positive diagonal
entries can be relaxed (see Remark~\ref{r-setup}).

The main ingredient of our work is that, based on the observation of a
relation between consensus  and  stability, the original consensus
problem on a sequence is reduced to the stability problem on a projected
sequence in a subspace. As a result, we can focus our study on
the eigenspace structure of the  projection operator. As a by-product,
we offer a simple proof of consensus to  deterministic linear networks.

The rest of the paper is arranged as follows: We start with the
problem formulation in section 2, where a crucial result on the
relation between consensus and stability is presented. In
section 3, a simple proof of consensus on a deterministic sequence is
provided, which can be read independently for readers only interested
in the deterministic case. The main result, the necessary and sufficient
condition for consensus of a random network, is established in section
4. Finally, we conclude our investigation in section 5.

\section{Problem formulation}
In the first subsection,  a crucial result
(\thmref{t-equivalence}) on the equivalence of consensus and stability
in a subspace will be presented under a general setup of consensus
problem. This theorem can be applied to very general cases setup,
including nonlinear and random sequences, and plays an 
important role throughout the paper. In the second subsection, the main
consensus problem is formulated using a linear stochastic difference
equation.

Before proceeding,  let us recall some standard notations:
\begin{enumerate}
\item In (column) vector space $\mathbb{R}^N$,  $x_i$ is the $i$th coordinate
  of vector $x\in \mathbb{R}^N$; $l^p$-norm is $\|x\|_p = \sum_{i=1}^N
  |x_i|^p,$  $ \forall 1\le p \le \infty$; $x^T$ denotes the transpose
  of $x$.
\item In square real matrix space
  $\mathbb{R}^{N\times N}$,  $I$ is the identity matrix; for all $A\in
  \mathbb{R}^{N\times N}$, $\|A\|_p = \max_{\|x\| = 1} \|Ax\|_p$ for all
  $1\le p \le \infty$; the eigenvalues will be arranged in order of
  $|\lambda_1(A)| \ge |\lambda_2(A)| \ge \cdots \ge |\lambda_N(A)|$;
  the spectral radius refers to  $\rho(A) = |\lambda_1(A)|$.   
\item $\|\cdot\|$ is used in the formula if it is valid for all
  $l^p$-norms.  
\item Given a probability space $(\Omega, \mathcal{F}, \mathbb{P})$,
  we denote by $\mathbb{E} = \mathbb{E}^{\mathbb{P}}$ the expectation
  under $\mathbb{P}$. $L^p$ refers to $L^p(\Omega, \mathcal{F},
  \mathbb{P})$: for random vector $Y: \Omega\to \mathbb{R}^N$, the
  $L^p$-norm is $\|Y\|_{L^p} = (\int_\Omega \|Y(\omega)\|_2
  \mathbb{P}(d\omega))^{1/p} = (\mathbb{E} [ (\|Y\|_2)^p ])^{1/p}.$
\end{enumerate}

\subsection{ A general consensus problem}

Let $(\Omega, \mathcal{F}, \mathbb{P}, \mathbb{F})$ be a filtered
probability space, where  $\mathbb{F} = \{\mathcal{F}_t: t = 0, 1, 2, 
\ldots\}$ is a sequence of increasing $\sigma$-algebras with
$\mathcal{F}_\infty \subset \mathcal{F}$. We consider an
$\mathbb{F}$-adapted sequence $\{X(t)\}$ taking values in
$\mathbb{R}^N$. In other words, $X(t)$ is a measurable mapping from
$(\Omega, \mathcal{F}_t) \to (\mathbb{R}^N,
\mathcal{B}(\mathbb{R}^N))$, where $\mathcal{B}(\mathbb{R}^N)$ is the
Borel $\sigma$-algebra on $\mathbb{R}^N$. Such a sequence
includes the general form of $$X(t) = f_t(X(t-1), \ldots, X(1))$$  
for some measurable function $f_t$, and  emphasizes its independence
of future events
First, we start from the precise definition of consensus on random
sequence in three different modes. As usual, $X(t, \omega)$ will be
used instead of $X(t)$ when we need to emphasize its
dependence on a sample path $\omega \in \Omega$. 

\begin{defn}
  [Consensus of a sequence]\label{d-3modes}
  Let $\{X(t)\}$ be an $\mathbb{F}$-adapted $\mathbb{R}^N$-valued
  random sequence. $\{X(t)\}$ is said to reach consensus
  \begin{enumerate}
  \item in probability, if 
    $$\lim_{t\to \infty} \mathbb{P} \Big\{\omega\in \Omega:  \max_{1\le i,j \le N}
    |X_i(t, \omega) - X_j(t,\omega)| > \varepsilon \Big\} = 0, \ \forall
    \varepsilon>0  .$$ 
  \item
    almost surely  $($with probability $1)$, if 
     $$\mathbb{P} \Big\{\omega\in \Omega: \lim_{t\to \infty}
     \max_{1\le i,j \le N} |X_i(t, \omega) - X_j(t,\omega)| = 0 \Big\}
     = 1.$$  
  \item
    in $L^p$  $(p\ge 1)$, if 
    $$\lim_{t\to \infty} \mathbb{E} \Big[\max_{1\le i,j \le N}
    |X_i(t) - X_j(t)|^p \Big] = 0.$$ 
  \end{enumerate}
\end{defn}

\begin{defn}
  [Stability of a sequence]
  $\{X(t)\}$ is said to be stable $($at zero$)$
  \begin{enumerate}
  \item in probability, if $$ \lim_{t\to \infty} \mathbb{P} \Big\{\omega
    \in \Omega:  
    \|X(t, \omega)\|> \varepsilon \Big\} = 0, \ \forall \varepsilon>0.$$
  \item
    almost surely,
    if $$\mathbb{P} \Big\{\omega\in \Omega: \lim_{t\to \infty} \|X(t,
    \omega)\| = 0 \Big\}   = 1.$$ 
  \item
    in $L^p$  $(p\ge 1)$, if  $$ \lim_{t\to \infty} \mathbb{E} \Big[\|X(t,
    \omega)\|^p \Big] = 0.$$
  \end{enumerate}
\end{defn}

Define a subspace of $\mathbb{R}^N$ by $\mathbb{R}_0 = \{x\in
\mathbb{R}^N: x_1 = x_2 = \ldots = x_N\}.$ 
Let  $\Pi$ be a projection operator on $\mathbb{R}_0$, i.e.
$$\Pi x = \langle x, v_0\rangle v_0, \quad  \forall x\in \mathbb{R}_0,$$ 
where $\langle\cdot, \cdot\rangle$ is  inner product,  $v_0\in
\mathbb{R}_0$ is an $l^2$-norm unit vector. We thus have  
the orthogonal projection $\Pi^\perp : \mathbb{R}^N \to
\mathbb{R}^\perp$ by $\Pi^\perp = I - \Pi$, so that the orthogonal
decomposition  is valid
\begin{equation}
  \label{eq:pip}
  x = \Pi x + \Pi^\perp x, \ \forall x\in \mathbb{R}^N.
\end{equation}
The following theorem shows that the consensus of a sequence in
$\mathbb{R}^N$ is equivalent to the
stability of the sequence projected on the subspace
$\mathbb{R}_0^\perp$.  
\begin{thm}
  \label{t-equivalence}
$\{X(t)\}$ reaches consensus almost surely $($respectively, in probability,
or in $L^p)$ if and only if $\{\Pi^\perp X(t)\}$ is
stable almost surely $($respectively, in probability,
or in $L^p)$.
\end{thm}
\begin{proof} We will show the equivalence of stability and consensus
  in the sense of 
  almost surely. The equivalence in probability and in $L^p$ can be
  similarly proved.
  \begin{enumerate}
  \item [($\Longrightarrow$)] Suppose $\{X(t)\}$ reaches consensus
    almost  surely. Define $Y(t)\in \mathbb{R}_0$ be a vector with all
    entries equal to the value of first coordinate of $X(t)$,
    i.e. $Y(t) = (X_1(t), X_1(t), \ldots, X_1(t))^T.$  Since  
    \begin{equation*}
      \begin{array}{ll}
        \| \Pi^\perp X(t)\|_\infty & = \min_{y\in \mathbb{R}_0} \|X(t)
        - y\|_\infty \\
        & \le \|X(t) - Y(t) \|_\infty \\
        & = \max_{1\le i \le N} |X_i(t) - X_1(t)| \\
        & \le \max_{1\le i,j \le N} |X_i(t) - X_j(t)| \to 0 
      \end{array}
    \end{equation*}
    as $t\to \infty$ almost surely. 
    Therefore, $\{\Pi^\perp X(t)\}$ is stable almost surely.
  \item [($\Longleftarrow$)] Suppose $\{\Pi^\perp X(t)\}$ is stable
    almost surely. Let $(\Pi X)_i(t)$ is the $i$th coordinate
    of vector $\Pi X(t)$. Note that, since $\Pi X(t) \in \mathbb{R}_0$,
    we have all coordinates with the same value, that is, $(\Pi
    X)_i(t) = (\Pi X)_j(t), \forall i, j$.  Therefore, by triangle
    inequality,  
    \begin{equation*}
      \begin{array}{ll}
        \max_{i,j} |X_i(t) - X_j(t)| & \le \max_{i,j} (|X_i(t) - (\Pi X)_i(t)|
        + |(\Pi X)_j(t) - X_j(t)|) \\ & \le \max_i |X_i(t) - (\Pi
        X)_i(t)| + \max_j |(\Pi X)_j(t) - X_j(t)| \\ & 
        \le  2 \|X(t) - \Pi X(t)\|_\infty \\
        & =  2 \|\Pi^\perp X(t)\|_\infty\to 0
      \end{array}
    \end{equation*}
    as $t\to \infty$ almost surely.  Therefore,  $\{X(t)\}$ reaches
    consensus almost surely. 
  \end{enumerate}
\end{proof}

\subsection{Consensus problem over linear random network}
We consider a similar setting as \cite{TJ08}. Let the space of $N\times N$
stochastic matrices be
\begin{equation}
  \label{eq:matrixspace}
  S_N = \Big\{A = (a_{ij})_{N\times N}: a_{ij} \ge 0, \  \sum_{j = 1}^N a_{ij}
  =1, \forall i,j  \Big\}
\end{equation}
and $\mathcal{B}(S_N)$ be the Borel $\sigma$-algebra on $S_N$. Let $\mu$
be a given probability distribution on $(S_N, \mathcal{B}(S_N))$, and
$\{A(t)\}$ be an $S_N$-valued i.i.d. sequence with distribution $\mu$. 
Let $\Omega = (S_N)^\infty$, $\mathbb{P} = \mu \times \mu \times
\cdots$,  and $\mathcal{F}_0 = \{\emptyset, \Omega\}$,  $\mathcal{F}_t =
\sigma(A(1), A(2), \ldots, A(t))$ for $t\ge 1$,
$\mathbb{F} = \cup_{t  = 0}^{\infty} \mathcal{F}_t $. 

Now, we consider a random sequence $\{X(t)\}$ given by
\begin{equation}
  \label{eq:X}
  X(t) = A(t) X(t-1), \ \forall t \in \mathbb{N}; \ X(0) = x.
\end{equation}
Then,  $\{X(t)\}$ is an $\mathbb{F}$-adapted sequence in the filtered
probability space $(\Omega, \mathcal{F}, \mathbb{P}, \mathbb{F})$.  
Observe that the distribution of $X(t)$ is determined by the initial
state $X(0) = x$ and distribution $\mu$.  Sometimes, we write $X^x(t)$
to emphasize the initial state $X(0) =x$ in the context. 

\begin{defn}
  [Consensus and stability of a distribution]
  \label{d-consensusd}
  A distribution $\mu$ is said to reach consensus
  almost surely
  $($respectively, in probability, or in $L^p)$, if $\{X^x(t)\}$ of
  \eqref{eq:X} generated by the distribution $\mu$ reaches consensus
  almost surely $($respectively, in probability, or in $L^p)$ for all
  initial states  $x\in \mathbb{R}^N$. 
\end{defn} 
Similar to \defnref{d-consensusd}, one can define stability for the
sequence \eqref{eq:X} generated by the distribution $\mu$.
\begin{rem}\label{r-setup}
  {\rm
    \cite{TJ08} had a different problem formulation in that the space
    $S_N$ of \eqref{eq:matrixspace} was replaced by a smaller space
    $$\hat S_N = \{ A\in S_N: \hbox{ all diagonal entries are strictly
      positive }\}.$$ 
    Indeed,  such a restriction is crucial in the proof of
    \cite[Theorem 3]{TJ08} to utilize the 
    \cite[Perron-Frobenius theorem]{BP94} on primitive matrix.   In
    our work, the diagonal entries 
    can be zero, which therefore covers  the results of \cite{TJ08} as
    a special case with $\mu(S_N\setminus   \hat S_N) =0$. 
    \qed
  }
\end{rem}

Next, our goal is to find a necessary and sufficient condition for the
consensus of the distribution $\mu$. 

\section{Necessary and sufficient condition for a deterministic
  sequence} 
A deterministic system can be treated as a special case of a random
system in the following sense. Let the probability distribution $\mu$
on $S_N$ satisfy $\mu(\{ A \}) = 1$ for some stochastic matrix $A \in
S_N$. Then,  $A(t) = A$ for all $t= 1,2, \ldots$, and the
sequence $\{X(t)\}$ of  \eqref{eq:X} becomes deterministic, and in
this case we have
$$X(t) = A^t x, \ \forall t = 0, 1, 2, \ldots$$
For convenience, we say $A$ reaches consensus if $\{X(t) = A^t x\}$
reaches consensus for all initial states $x\in \mathbb{R}^N$. Note
that, the deterministic consensus is indifferent to all three modes of
consensus, since the sample space $\Omega$ can be treated as a singleton
$\{A\}\times \{A\} \times \cdots$. Thus, this definition is consistent
with  \defnref{d-consensusd} on consensus (in  all three modes) of
distribution $\mu$  of the form $\mu(\{A\}) = 1$.  

The main idea in this section is that, thanks to
\thmref{t-equivalence}, it is equivalent to find a sufficient and
necessary condition of the stability of $\{\Pi^\perp X(t)\}$, which
turns out to be a sequence generated by the projection matrix
$\Pi^\perp A$. We will show that  $\Pi^\perp A$ is stable if and
only if $\rho(\Pi^\perp A)<1$ by \propref{p-Stable}. Together with the
fact that $\rho(\Pi^\perp A) = |\lambda_2(A)|$, by \propref{p-Jordan},
we will obtain the desired   necessary and sufficient  condition.

To proceed with the consensus on the deterministic sequence $X(t)$
generated by the stochastic matrix $A$,  we first recall  some
properties of 
stochastic matrices. Since each row sum of a stochastic matrix is
equal to 
$1$, its largest eigenvalue is $ \rho(A) = \lambda_1(A) =1$, i.e. $A x
= x$ for all $x\in \mathbb{R}_0$. Also, we have
\begin{equation}
  \label{eq:inftynorm}
  \|A^t x\|_\infty  \le \|x\|_\infty, \ \forall x\in \mathbb{R}^N, t\in
  \mathbb{N}.
\end{equation}
In addition, we have the following
useful results:
\begin{prop}
  \label{p-Jordan}
  Let $A\in S_N$ be a stochastic matrix. Then 
  \begin{enumerate}
  \item $A$ has a Jordan canonical form of
    \begin{equation}
      \label{eq:p-Jordan5}
      \Lambda = \left[
      \begin{array}{ll}
        1 & 0_{1\times (N-1)} \\
        0_{(N-1) \times 1} & \Lambda_{22}
      \end{array}\right],
    \end{equation} where $0_{m\times n}$ is $m\times n$ matrix
    with each entry being zero, and  $\Lambda_{22}$ is sub-matrix of
    Jordan form.
  \item The linear operator $\Pi^\perp$ defined in \eqref{eq:pip}
    satisfies 
    \begin{equation}
      \label{eq:p-Jordan1}
      \Pi^\perp A = \Pi^\perp A \Pi^\perp
    \end{equation}
  \item $\Pi^\perp A$, as a matrix, has a Jordan form of 
    $\Lambda_0 = \left[ 
      \begin{array}{ll}
        0 & 0 \\
        0 & \Lambda_{22}
      \end{array}\right]$ with $\Lambda_{22}$ defined in
    \eqref{eq:p-Jordan5}. In particular, $\rho(\Pi^\perp A) =
    |\lambda_2 (A)|$.
  \end{enumerate}
\end{prop}
\begin{proof}
  \begin{enumerate}
  \item Let $v_0$ be the unit vector in the subspace $\mathbb{R}_0$.
    Note that $v_0 \in \mathbb{R}_0$ is an eigenvector of $A$
    associated with the eigenvalue $\lambda_1 (A)= 1$, i.e. $(A -I) v_0
    = 0$. To prove the first claim, we only need to show that the
    Jordan block corresponding to $\lambda_1(A) =1$ is simple. If not,
    there exists $v_1 \notin \mathbb{R}_0$ associated with the Jordan
    block 
    $\left[ 
      \begin{array}{ll}
        1 & 1 \\
        0 & 1
      \end{array}\right]$, satisfying
    \begin{equation}
      \label{eq:p-Jordan4}
      (A - I)v_1 = v_0. 
    \end{equation}
    By induction, 
    $$A^t v_1 = t \cdot v_0 + v_1, \ \forall t\in \mathbb{N}.$$
    This implies that $\|A^t v_1\|_\infty \to \infty$ as $t\to \infty$,
    which leads to a contradiction to \eqref{eq:inftynorm}.
  \item
    One can prove \eqref{eq:p-Jordan1} as follows: $\forall x\in
    \mathbb{R}^N$ 
    $$\Pi^\perp A x = \Pi^\perp A (\Pi x + \Pi^\perp x) = \Pi^\perp A
    \Pi x + \Pi^\perp A \Pi^\perp x = \Pi^\perp \Pi x + \Pi^\perp A
    \Pi^\perp x  = \Pi^\perp A \Pi^\perp x.  $$
  \item
    If $\{x_1, \ldots x_m\}$ are generalized eigenvectors of $A$
    associated with some eigenvalue $\lambda$ in the Jordan block in
    $\Lambda_{22}$, satisfying   
    $$(A - \lambda I)x_i = x_{i-1}, \hbox{ for
    } i = 1,2, \ldots m, \quad x_0 = 0,$$ 
    then, by the facts $x_1\notin \mathbb{R}_0$ and \eqref{eq:p-Jordan1},
    $$
    \begin{array}{ll}
      (\Pi^\perp A - \lambda I)(\Pi^\perp x_i) & = \Pi^\perp A
      \Pi^\perp x_i -  \lambda \Pi^\perp x_i\\ 
      & = \Pi^\perp A x_i -  \lambda \Pi^\perp x_i
      \\ 
      & = \Pi^\perp ( A -  \lambda I) x_i  \\ 
      & = \Pi^\perp x_{i-1}.
    \end{array}
    $$
    In other words, since $\Pi^\perp x_1 \neq 0$, $\Pi^\perp A$ with
    $\{\Pi^\perp x_1, \ldots, \Pi^\perp x_m\}$ preserves the structure
    of the eigenspace associated with  matrix $A$ corresponding to the
    eigenvalue $\lambda$ in Jordan    block $\Lambda_{22}$. Also, since 
    $\lambda$ is arbitrary eignevalue in the Jordan block $\Lambda_{22}$,
    together with $\Pi^\perp A v_0 = 0$, we conclude $\Pi^\perp A$ has a
    Jordan form of     $\Lambda_0 = \left[ 
      \begin{array}{ll}
        0 & 0 \\
        0 & \Lambda_{22}
      \end{array}\right]$. Finally, we have
    $$\rho(\Pi^\perp A) = \rho(\Lambda_0) = \rho(\Lambda_{22}) =
    |\lambda_2(A)|.$$ 
  \end{enumerate} 
\end{proof}

Next, we present a necessary and sufficient condition for stability.
\begin{prop} \label{p-Stable}
  $A\in \mathbb{R}^{N\times N}$ is stable if and only if   $ \rho(A) < 1$.
\end{prop}
\begin{proof}One can use the fact $\lim_{t\to \infty}\|A^t\|^{1/t} =
  \rho(A)$ to complete the proof.
\end{proof}

Thanks to \propref{p-Jordan} and \propref{p-Stable}, we are now ready
to obtain a necessary and sufficient condition for the consensus of a
deterministic sequence. 
\begin{thm}   [Necessary and sufficient condition in deterministic
  case]
  \label{t-detconsensus}
  $A \in S_N$ reaches consensus if and only if $|\lambda_2(A)| <1$.
\end{thm}
\begin{proof}
  By \thmref{t-equivalence}, $A$ reaches consensus if and only if
  $\{\Pi^\perp X(t)\}$ is stable. Note that, by \eqref{eq:p-Jordan1},
  for any initial state $x \in \mathbb{R}^N$ 
  \begin{equation}
    \label{eq:t-detconsensus1}
    \begin{array}{ll}
      \Pi^\perp X(t) & = \Pi^\perp A X(t-1) = (\Pi^\perp A) \Pi^\perp
      X(t-1) \\ & = \cdots =  (\Pi^\perp A)^t \Pi^\perp x = (\Pi^\perp A)^t
      x.
    \end{array}
  \end{equation}
  Thus, $\{\Pi^\perp X(t)\}$ is a sequence generated by $\Pi^\perp
  A$. By \propref{p-Stable},   $\{\Pi^\perp X(t)\}$ is stable if and
  only if $\rho(\Pi^\perp A) <1$. Observe that, by \propref{p-Jordan},
  $\rho(\Pi^\perp A) = |\lambda_2 (A)|$. This completes the proof. 
\end{proof}

\section{Necessary and sufficient condition for a stochastic sequence} 
In this section, we return to the stochastic sequence $\{X(t)\}$
defined in \eqref{eq:X} generated by distribution $\mu$, and study
a necessary 
and sufficient condition for its consensus. First, by studying the fine
structure of the random sequence 
generated by i.i.d. stochastic matrices, we show that consensus in three
different modes classified by \defnref{d-3modes}  are in fact
equivalent to each other. Thus, we can only work on the almost surely
consensus.  

\subsection{Equivalence of consensus in three modes}

Before we proceed with the equivalence of consensus in three modes, we 
briefly recall some relations between convergence of random variables in
three modes, and we refer to  \cite{Dur05} for more detail. Consider a
sequence of random variables $\{a_n, n =1, 2, \ldots\}$
and a random variable $a \ge 0$. Both almost surely
convergence and $L^1$ convergence 
imply in probability convergence, i.e. $a_n \to a$ almost surely
implies $a_n \to a$ in probability; $a_n \to a$ in $L^1$ implies $a_n
\to a$ in probability. However, the reverse directions need further
conditions in general. $a_n \to a$ in probability together with $|a_n|
\le |b|$ almost surely for some $b\in L^1(\Omega, \mathcal{F},
\mathbb{P})$ implies $a_n \to a$ in $L^1$ by the dominated convergence
theorem; $a_n \to a$ in $L^1$ and $0\le a_{n+1} \le a_n$ almost surely
together implies $a_n \to a$ almost surely by the monotone convergence
theorem. 

\begin{lem}
  \label{l-modes}
  Consider  the sequence $\{X(t)\}$ defined in \eqref{eq:X} generated
  by 
  distribution $\mu$. Given $X(0) = x$, the following statements on
  stability of $\{\Pi^\perp X(t)\}$  are equivalent:
  \begin{enumerate}
  \item $\{\Pi^\perp X(t)\}$ is stable in probability.
  \item
     $\{\Pi^\perp X(t)\}$ is stable in $L^1$.
   \item
      $\{\Pi^\perp X(t)\}$ is stable almost surely.
  \end{enumerate}
\end{lem}
\begin{proof} Observe that, by \eqref{eq:p-Jordan1}, $\{\Pi^\perp
  X(t)\}$ is a sequence  generated by random matrix $\Pi^\perp A(t)$,
  i.e. 
  \begin{equation}
    \label{eq:l-modes1}
    \Pi^\perp X(t) = \Pi^\perp A(t) \ \Pi^\perp X(t-1).
  \end{equation}
  In the following, we prove the equivalence by showing: (1) implies
  (2), (2) implies (3), (3) implies (1), respectively.
  \begin{enumerate}
  \item If the sequence $\{\Pi^\perp X(t)\}$ is stable in probability,
    then $\|\Pi^\perp X(t)\|_\infty \to 0$ in probability.  Together
    with the uniform boundedness  $\|\Pi^\perp X(t)\|_\infty \le
    \|x\|_\infty$, the dominated convergence theorem implies that
    $\|\Pi^\perp X(t)\|_\infty \to 0$  in $L^1$. Thus, the sequence
    $\{\Pi^\perp X(t)\}$ is stable in $L^1$. 
  \item
     If the sequence $\{\Pi^\perp X(t)\}$ is stable in $L^1$, then
     $\|\Pi^\perp X(t)\|_\infty \to 0$ in $L^1$.   In addition, one
     can show the monotonicity of  $\|\Pi^\perp X(t)\|_\infty \le
     \|\Pi^\perp X(t-1)\|_\infty$, by observing
     \begin{equation}
       \label{eq:monoton}
       \|\Pi^\perp X(t)\|_\infty = \|\Pi^\perp A(t)X(t-1)\|_\infty  =
       \|\Pi^\perp A(t) \Pi^\perp X(t-1)\|_\infty \le
       \| \Pi^\perp X(t-1)\|_\infty.   
     \end{equation}
     By the monotone convergence theorem, $\|\Pi^\perp X(t)\|_\infty
     \to 0$  almost surely. 
  \item
    It is well known that  almost surely convergence implies
    convergence in probability. 
  \end{enumerate}
\end{proof}

The next theorem about equivalent consensus in three modes is a main
result in our paper.
\begin{thm}
  \label{t-modes}
  Consider  the sequence $\{X(t)\}$ defined in \eqref{eq:X} generated by
  distribution $\mu$. The following statements on consensus are
  equivalent: 
  \begin{enumerate}
  \item Distribution $\mu$ reaches consensus in probability.
  \item
    Distribution $\mu$  reaches consensus   in $L^1$.
  \item
    Distribution $\mu$  reaches consensus  almost surely.
  \end{enumerate}
\end{thm}
\begin{proof}
  It follows from  \thmref{t-equivalence} and
  \lemref{l-modes}.
\end{proof}

\subsection{Necessary and sufficient condition for the random case} 
Thanks to \thmref{t-modes}, our work is now reduced to finding a
necessary and sufficient condition for consensus in any one of the
three modes. Below, we say the distribution $\mu$ reaches consensus
without specifying a convergence mode.

Recall from  \defnref{d-consensusd} that, a distribution $\mu$ reaches  
consensus if the generated sequence $\{X^x(t)\}$ reaches consensus
for all initial states $x\in \mathbb{R}^N$.   The next
proposition shows that it is sufficient to check the consensus of
$\{X^x(t)\}$ only for all $x\in (\mathbb{R}^+)^N$ to guarantee the
consensus of a distribution $\mu$.  
\begin{prop}
  \label{p-positiveX}
  $\mu$ reaches consensus if and only if $X^x(t)$ defined in
  \eqref{eq:X} reaches consensus for all $x\in (\mathbb{R}^+)^N$.
\end{prop}
\begin{proof} Observe that $X^{x+c}(t) = X^x(t) + c$ for all $c\in
  \mathbb{R}_0$. Hence,
  $$\max_{1\le i,j \le N}  |X_i^x(t, \omega) - X_j^x(t,\omega)| =
  \max_{1\le i,j \le N}  |X_i^{x+c}(t, \omega) - X_j^{x+c}(t,\omega)|,
  \ \forall c\in \mathbb{R}_0.$$ 
  In other words, to consider consensus of $\{X^x(t)\}$ for some $x
  \notin (\mathbb{R}^+)^N$, one can always investigate the consensus
  of $X^{x + c}(t)$ equivalently, by taking  $c = (\|x\|_\infty,
  \ldots, \|x\|_\infty)^T \in \mathbb{R}_0$. Note $x + c \in
  (\mathbb{R}^+)^N$,  hence the result holds. 
\end{proof}

Next, we review some useful properties of the expectation operator
$\mathbb{E}$. First, the expectation operator $\mathbb{E}$ is
commutative with any deterministic matrix $A$, i.e.,
\begin{equation}
  \label{eq:comm}
  A \mathbb{E} [Y] =  \mathbb{E} [A Y], \ \forall
  \mathcal{F}\hbox{-measurable } Y:  \Omega \to \mathbb{R}^N.
\end{equation}
In particular, by taking $A = \Pi^\perp$, we have $\Pi^\perp
\mathbb{E} = \mathbb{E} \Pi^\perp$. Furthermore, for an arbitrary
random matrix $A$,  if a random vector $Y: \Omega\to \mathbb{R}^N$ is
independent of $A$, then 
\begin{equation}
  \label{eq:comm1}
   \mathbb{E} [A Y] = \mathbb{E}[A] \mathbb{E} [Y].
\end{equation}
Finally, note that the deterministic sequence $\{\mathbb{E}[X(t)]\}$
is actually a sequence generated by the deterministic matrix
$\mathbb{E}[A]$, 
since by  \eqref{eq:comm1}
\begin{equation}
  \label{eq:EX}
  \begin{array}{ll}
    \mathbb{E}[X(t)] & = 
    \mathbb{E}[A(t)X(t-1)] \\
    & = \mathbb{E}[A(t)] \mathbb{E}[X(t-1)] \\
    & = \mathbb{E}[A(1)]\cdot \mathbb{E}[X(t-1)].
  \end{array}
\end{equation}

\begin{thm}[Necessary and sufficient condition for consensus]
  \label{t-randconsensus} Consider  $\{X(t)\}$ defined in \eqref{eq:X}
  generated by distribution $\mu$. 
  $\mu$ reaches consensus almost surely (also, in probability, and in
  $L^1$) 
  if and only if $\lambda_2 (\mathbb{E}^\mu [A(1)]) <1$.  
\end{thm}
\begin{proof}
  By saying that $\mu$ reaches consensus, we mean $\mu$ reaches
  consensus in any of three modes due to \thmref{t-modes}.
  \begin{enumerate}
  \item [($\Longrightarrow$)] If $\mu$ reaches consensus, then
    \thmref{t-equivalence} implies
    $\Pi^\perp X(t) \to 0$ in $L^1$  for all
    initial  states $X(0) = x$, hence $\mathbb{E}[\|\Pi^\perp X(t)\|_\infty]
    \to 0$.

    Next, by  \eqref{eq:comm} and Jensen's    inequality, we have
    $$\| \Pi^\perp \mathbb{E}[ X(t)]\|_\infty = \|\mathbb{E}[\Pi^\perp
    X(t)]\|_\infty \le \mathbb{E}[\|\Pi^\perp X(t)\|_\infty]  \to 0 $$
    This implies that the deterministic sequence $\{\Pi^\perp
    \mathbb{E}[ X(t)]\}$ is stable. 
    Then,  using \eqref{eq:EX} and \eqref{eq:p-Jordan1}, we have
    $$\Pi^\perp \mathbb{E}[X(t)] = \Pi^\perp \mathbb{E}[A(1)]\cdot 
    \mathbb{E}[X(t-1)]   = \Pi^\perp \mathbb{E}[A(1)]\cdot \Pi^\perp
    \mathbb{E}[X(t-1)] $$
    In other words, $\{\Pi^\perp
    \mathbb{E}[ X(t)]\}$  is a deterministic sequence generated by  
    matrix $\Pi^\perp \mathbb{E}[A(1)]$. Thus, by
    \propref{p-Jordan} and  \propref{p-Stable}, $\rho(\Pi^\perp
    \mathbb{E}[A(1)]) = |\lambda_2(\mathbb{E}[A(1)])| < 1$.   
  \item[($\Longleftarrow$)]
    It follows from \eqref{eq:EX} that the deterministic sequence
    $\{\mathbb{E}[X(t)]\}$  is  generated by matrix
    $\mathbb{E}[A(1)]$. 
    If $|\lambda_2(\mathbb{E}[A(1)])| < 1$, then by applying
    \thmref{t-detconsensus} on \eqref{eq:EX}, we conclude that
    the sequence $\{\mathbb{E}[X(t)]\}$ reaches consensus. Hence, the
    deterministic sequence 
    $\{\Pi^\perp\mathbb{E}[X(t)] =  \mathbb{E}[\Pi^\perp X(t)] \}$
    is stable by \thmref{t-equivalence}, i.e., $ \mathbb{E} [\Pi^\perp X(t)] \|_1
      \to 0.$ 
    By \propref{p-positiveX}, we can always assume $x\in
    (\mathbb{R}^+)^N$. Thus, $\Pi^\perp X(t) \in (\mathbb{R}^+)^N$,
    and this leads to 
    \begin{equation}
      \label{eq:t-randconsensus1}
      \mathbb{E} [\|\Pi^\perp X(t)\|_\infty] \le \mathbb{E}
      [\|\Pi^\perp X(t)\|_1] = \| \mathbb{E} [\Pi^\perp X(t)] \|_1
      \to 0.
    \end{equation}
    In other words, $\{\Pi^\perp X(t)\}$ is stable in $L^1$. This
    implies  the consensus of $\mu$ by \thmref{t-equivalence}.      
  \end{enumerate}
\end{proof}
\begin{rem}
  {\rm
    In \eqref{eq:t-randconsensus1}, we used the fact, for all $Y:
    \Omega \to (\mathbb{R}^+)^N$, 
    $$\mathbb{E}[\|Y\|_1] = \mathbb{E} \Big[\sum_{i=1}^N Y_i \Big] =
    \sum_{i=1}^N \mathbb{E}[Y_i] = \|\mathbb{E} Y\|_1.$$
    However, one shall not expect identity $\mathbb{E}[\|Y\|_\infty ]
    = \|\mathbb{E}[Y]\|_\infty$ holds in general. For instance, one
    has strict a inequality, if 
    $Y(\omega_1) = (0,1)^T$,  $Y(\omega_2) = (1, 0)^T$, and
    $\mathbb{P}(\{\omega_1\}) = \mathbb{P}(\{\omega_2\}) = 1/2$, then
    $\mathbb{E}[\|Y\|_\infty ] = 1 > \frac 1 2 =
    \|\mathbb{E}[Y]\|_\infty$ holds. This is the reason that we use
    the $l^1$-norm in \eqref{eq:t-randconsensus1} instead of directly
    but incorrectly
    using $ \mathbb{E}  [\|\Pi^\perp X(t)\|_\infty] =  \| \mathbb{E}
    [\Pi^\perp X(t)] \|_\infty   \to 0. $ 
  }
\end{rem}

\begin{cor}
  \label{c-equiv}
  Consider  $\{X(t)\}$ defined in \eqref{eq:X} generated by distribution
  $\mu$.  $\mu$  reaches consensus if and only if
  $\mathbb{E}^\mu[A(1)]$ reaches   consensus. 
\end{cor}

\section{Concluding remarks}
In this paper, we derived a necessary and sufficient condition for
consensus over a linear random network based on the connection between
consensus and stability as shown by \thmref{t-equivalence}. Although
our proof is shown under the discrete-time framework, the similar
results still hold in the continuous-time setting. Consequently, one
can similarly follow the procedure to obtain consensus conditions based
on stability results of \cite{ZYS09} on hybrid switching
continuous-time systems.  

Regarding the second-order random network, one can also utilize the result
of this work. More precisely, for 
$$X(t) = \alpha A(t) X(t-1) + \beta B(t) X(t-2), \quad t = 2,3,
\ldots$$
where $\alpha + \beta =1$, $\alpha, \beta \ge 0$, $\{A(t)\}$ and
$\{B(t)\}$ are i.i.d. stochastic matrix sequence with given distributions
$\mu_A$ and $\mu_B$ on $S_N$, the problem is equivalent to 
$$Y(t) = C(t) Y(t-1),$$
where $Y(t) = \left[
  \begin{array}{l}
    X(t)\\
    X(t-1)
  \end{array}\right]$ is an $\mathbb{R}^{2N}$-vector and $C(t) = \left[
  \begin{array}{ll}
    \alpha A(t) & \beta B(t) \\ I & 0
  \end{array} \right]$ is a stochastic matrix in $S_{2N}$.

One last interesting remark is the application of Kolmogorov's 0-1 law
\cite{Dur05}, which is firstly given in \cite{TJ08} in the context of
ergodicity of i.i.d. matrix sequence. Similar result also holds 
in the consensus and the stability problems. For instance, now we know that
$X^x(t)$ defined in \eqref{eq:X} does not reach consensus when
$\lambda_2(\mathbb{E} [A(1)])\ge 1$ for a given distribution $\mu$ in
any of the three modes. In other words,  
$$\mathbb{P} \Big\{\omega\in \Omega: X^x(t, \omega) \hbox{ reaches
  consensus for all } x\in \mathbb{R}^N \Big\} < 1.$$  
Natural question is then, what the above probability is when
$\lambda_2(\mathbb{E} [A(1)])\ge 1$. The answer is surprisingly
simple: zero.
\begin{prop}
  Consider $X^x(t)$ defined in \eqref{eq:X} generated by i.i.d. matrices
  with distribution $\mu$. Then
  $$\mathbb{P} \Big\{\omega\in \Omega: X^x(t, \omega) \hbox{ reaches
    consensus for all } x\in \mathbb{R}^N \Big\} $$  
  is either $1$ or $0$.
\end{prop}
Proof can be accomplished similarly to \cite[Lemma 1]{TJ08}, by 
using the tail $\sigma$-field argument on decreasing events of the form
$$B_k = \Big\{ \omega: \Pi_{t=k}^\infty A(t) x \hbox{ reaches consensus for all }
x\in \mathbb{R}^N\Big\}.$$

\bibliographystyle{plain}

\end{document}